\documentclass[a4paper,10pt]{article}
\usepackage{amssymb,amsmath,amsthm}
\usepackage{fullpage}
\usepackage{hyperref}
\usepackage{eucal}
\usepackage{color}
\usepackage{stackrel}
\usepackage{graphicx}
\newcommand{\ie}{\emph{i.e.}}
\newcommand{\eg}{\emph{e.g.}}
\newcommand{\cf}{\emph{cf.}}
\newcommand{\Real}{\mathbb{R}}
\newcommand{\Com}{\mathbb{C}}

\newcommand{\Int}{\mathbb{Z}}
\newcommand{\Sphere}{\mathbb{S}}

\newcommand{\supp}{\mathop{\mathrm{supp}}\nolimits}
\newcommand{\Dom}{\mathop{\mathsf{dom}}\nolimits}
\newcommand{\dist}{\mathop{\mathrm{dist}}\nolimits}
\newcommand{\rot}{\mathop{\mathrm{rot}}\nolimits}
\newcommand{\divergence}{\mathop{\mathrm{div}}\nolimits}

\newcommand{\eps}{\varepsilon}
\newcommand{\sii}{L^2}

\newcommand{\der}{\mathrm{d}}
\newtheorem{Theorem}{Theorem}
\newtheorem{Lemma}{Lemma}

\theoremstyle{definition}
\newtheorem{Remark}{Remark}

\usepackage[normalem]{ulem}
\definecolor{DarkGreen}{rgb}{0,0.5,0.1} 

\newcommand\soutD{\bgroup\markoverwith
{\textcolor{DarkGreen}{\rule[.5ex]{2pt}{1pt}}}\ULon}
\newcommand\soutP{\bgroup\markoverwith
{\textcolor{blue}{\rule[.5ex]{2pt}{1pt}}}\ULon}
\newcommand{\Hm}[1]{\leavevmode{\marginpar{\tiny%
$\hbox to 0mm{\hspace*{-0.5mm}$\leftarrow$\hss}%
\vcenter{\vrule depth 0.1mm height 0.1mm width \the\marginparwidth}%
\hbox to
0mm{\hss$\rightarrow$\hspace*{-0.5mm}}$\\\relax\raggedright #1}}}

\begin{document}
%
\title{\textbf{\Large
Complex magnetic fields: 
An improved Hardy-Laptev-Weidl inequality 
and quasi-self-adjointness 
}}
\author{David Krej\v{c}i\v{r}{\'\i}k}
\date{\small 
\emph{
\begin{quote}
\begin{center}
Department of Mathematics, Faculty of Nuclear Sciences and 
Physical Engineering, Czech Technical University in Prague, 
Trojanova 13, 12000 Prague 2, Czech Republic;
david.krejcirik@fjfi.cvut.cz.
\end{center}
\end{quote}
}
\smallskip
15 February 2018}
\maketitle
 
%
\begin{abstract}
\noindent
We show that allowing magnetic fields to be complex-valued
leads to an improvement in the magnetic Hardy-type inequality
due to Laptev and Weidl. 
The proof is based on the study of  
momenta on the circle with complex magnetic fields,
which is of independent interest in the context of  
$\mathcal{PT}$-symmetric and quasi-Hermitian quantum mechanics.
We study basis properties of the non-self-adjoint momenta
and derive closed formulae for the similarity transforms 
relating them to self-adjoint operators.
\end{abstract}
%

\section{Introduction} 
%
The magnetic field in quantum mechanics
is an appealing concept for several reasons.
First, it enters the theory indirectly, 
through its vector potential in the Schr\"odinger equation,
leading thus to purely quantum phenomena 
such as the Aharonov-Bohm effect.
Second, the magnetic Hamiltonian 
is not invariant under complex conjugation,
which results in the possibility of degeneracy of the ground state
and other abnormal properties with respect to the theory of 
elliptic partial differential equations with real coefficients.
Finally, the aesthetic dimension of the theory is also to be emphasised, 
for it can be described by a mathematically elegant, 
geometric framework of differential forms.
Instead of choosing a few works from the huge bibliography
devoted to the magnetic field in mathematical physics
over the last hundred years, we refer to the recent 
Raymond's monograph \cite{Raymond-book},
initially entitled ``Little magnetic book'',
with many references. 

In this paper we point out some new phenomena related 
to magnetic fields which are allowed to be complex-valued.
It is remarkable that there is an experimental evidence
for imaginary magnetic fields in quantum statistical physics 
\cite{Peng-etal_2015, Ananikian-Kenna_2015}.
However, our primary motivations come 
from more theoretical aspects in quantum mechanics that we explain now.

\subsection{Improved magnetic Hardy inequalities}
Hardy inequalities are sometimes interpreted 
as a sort of uncertainty principle in quantum mechanics.
More specifically, they are functional inequalities 
quantifying positivity of elliptic operators.
For instance, the spectrum of the free Hamiltonian in~$\Real^d$,
\ie\ the self-adjoint realisation of the Laplacian in $\sii(\Real^d)$,
is stable against small short-range perturbations if, and only if,
a Hardy inequality holds, that is, $d \geq 3$.   
It is a noteworthy observation of Laptev and Weidl~\cite{Laptev-Weidl_1999} 
that Hardy inequalities hold even in the two-dimensional situation
provided that magnetic fields are added.

Our first observation is that complex ``magnetic fields'' 
lead to improved variants of the Hardy-type inequality 
due to Laptev and Weidl~\cite{Laptev-Weidl_1999} in the real case. 

\begin{Theorem}\label{Thm.Hardy}
Let a smooth complex field $B:\Real^2\to\Com$ 
be compactly supported and
satisfy the condition
\begin{equation}\label{flux}
  \frac{1}{2\pi}\int_{\Real^2} \Re B(x) \, \der x \not\in \Int
  \qquad \mbox{or} \qquad
  \frac{1}{2\pi}\int_{\Real^2} \Im B(x) \, \der x \not=0
  \,.
\end{equation}
Then there exists a positive constant~$c$ depending on~$B$ 
such that the following inequality holds
\begin{equation}\label{Hardy}
  \forall \psi \in C_0^\infty(\Real^2) \,, \qquad
  \int_{\Real^2} |\nabla_{\!A}\psi(x)|^2 \, \der x
  \geq c \int_{\Real^2} \frac{|\psi(x)|^2}{1+|x|^2} \, \der x
  \,,
\end{equation}
where $\nabla_{\!A} := \nabla - iA$ with
\begin{equation}\label{gauge}
  A(x) := (-x_2,x_1) \int_0^1 B(tx) \, t \, \der t
  \,. 
\end{equation}
\end{Theorem}

In the sentence preceding the theorem, 
we have used quotation marks because the left-hand side of~\eqref{Hardy}
is a symmetric quadratic form, 
associated with a self-adjoint operator in $\sii(\Real^2)$,
even if the vector potential~$A$ is complex-valued.
In other words, the differential expression $(\nabla_{\!A})^{*}\,\nabla_{\!A}$
is of course formally self-adjoint, 
while the magnetic Hamiltonian would rather correspond 
to $-\nabla_{\!A}\nabla_{\!A}$,
which is non-self-adjoint whenever~$A$ is not real-valued.
Hence, Theorem~\ref{Thm.Hardy} is potentially relevant 
in the context of conventional quantum mechanics,
where the ``imaginary part of the magnetic field'' typically
relates the problem to a self-adjoint Schr\"odinger operator
in a weighted $\sii$-space (\cf~Remark~\ref{Rem.weight} below).

For $\Im B = 0$, inequality~\eqref{Hardy} is the celebrated result of 
Laptev and Weidl \cite[Thm.~1]{Laptev-Weidl_1999}
and~\eqref{flux} reduces to the standard flux condition.
The novelty here is that, for complex magnetic fields, 
one can get a non-trivial Hardy inequality
even if the first condition of~\eqref{flux} is not satisfied. 
Indeed, it is enough that the magnetic field
satisfies the second condition of~\eqref{flux},
which is a kind of complex extension of the standard flux condition.

The relationship~\eqref{gauge} ensures that 
$B = \rot A := \partial_1 A_2 - \partial_2 A_1$
and that the transverse condition
\begin{equation}\label{transverse}
  x \cdot A(x) = 0
  \,. 
\end{equation}
holds for all $x \in \Real^2$.
In view of the gauge invariance for real-valued~$B$, 
the requirement~\eqref{gauge} is superfluous
and~\eqref{Hardy} holds for any smooth vector potential~$A$ 
satisfying $B = \rot A$.
In the complex case, however, 
the choice~\eqref{gauge} is an inevitable part of the theorem.

The restriction to two-dimensional magnetic fields in Theorem~\ref{Thm.Hardy}
is just for simplicity of the presentation.
Higher-dimensional analogues of Theorem~\ref{Thm.Hardy}
can be derived quite straightforwardly 
by combining the ideas of the present paper with the methods of~\cite{CK}.
Notice also that for non-trivial real-valued~$B$ 
the condition~\eqref{flux} is not needed provided that 
the weight $(1+|x|^2)^{-1}$ on the right-hand side of~\eqref{Hardy}
is replaced by $(1+|x|^2\log^2|x|)^{-1}$, see~\cite{CK}.
In this paper complex extensions of this robust result
are also provided 
(see Theorems~\ref{Thm.Hardy.log} and~\ref{Thm.Hardy.robust} below).
On the other hand, Theorem~\ref{Thm.Hardy} is sharp in the sense
that~\eqref{Hardy} cannot hold with a positive constant~$c$
provided that the condition~\eqref{flux} is violated
(see Remark~\ref{Rem.optimal} below).

Finally, let us notice that a singular Hardy weight 
on the right-hand side of~\eqref{Hardy} is admissible
provided that the vector potential is singular, too.
As an illustration, we consider the complex Aharonov-Bohm potential
\begin{equation}\label{AB}
  A_\infty(x) := (-x_2,x_1) \, \frac{\alpha}{|x|^2}
  \,, \qquad \mbox{where} \qquad
  \alpha \in \Com
  \,. 
\end{equation}
It can be obtained from~\eqref{gauge} by formally putting $B = 2\pi \alpha \delta$,
where~$\delta$ is the Dirac delta function.

\begin{Theorem}\label{Thm.AB}
Let the complex vector potential $A_\infty:\Real^2 \to \Com^2$
given by~\eqref{AB} satisfy the condition
\begin{equation}\label{flux.AB}
  \Re\alpha \not\in \Int
  \qquad \mbox{or} \qquad
  \Im\alpha \not=0
  \,.
\end{equation}
Then there exists a positive constant~$c_\infty$ 
depending on~$\alpha$
such that the following inequality holds
\begin{equation}\label{Hardy.AB}
  \forall \psi \in C_0^\infty(\Real^2\setminus\{0\}) \,, \qquad
  \int_{\Real^2} |\nabla_{\!A_\infty}\psi(x)|^2 \, \der x
  \geq c_\infty \int_{\Real^2} \frac{|\psi(x)|^2}{|x|^2} \, \der x
  \,.
\end{equation}
\end{Theorem}

This theorem is a complex extension of another result 
of Laptev and Weidl \cite[Thm.~3]{Laptev-Weidl_1999} in the real case.
The new observation here is that 
the imaginary part of~$\alpha$ leads to positivity
even if its real part is an integer. 

The proofs of Theorems~\ref{Thm.Hardy} and~\ref{Thm.AB} follow 
by extending
the original ideas of~\cite{Laptev-Weidl_1999} to the complex case.
In some aspects we rather follow the more recent approach 
of~\cite{K7} and~\cite{CK}.
The main idea is to use the transverse condition~\eqref{transverse}
and, by employing polar coordinates, reduce the problem 
to considering a magnetic Schr\"odinger operator on the circle.
This leads us to another interesting aspect of complex-valued
magnetic fields in quantum mechanics.

\subsection{Quasi-self-adjointness of momenta with complex magnetic fields}
Many sustained attempts have been made by physicists in recent years
to represent observables in quantum mechanics
by operators which satisfy certain 
physical-like symmetries instead of the self-adjointness.
This ``extension'' of quantum mechanics is often referred 
to as \emph{$\mathcal{PT}$-symmetric quantum mechanics} 
and its advent is usually 
associated with the pioneering work of Bender and Boettcher
from 1998 \cite{Bender-Boettcher_1998}.
The quotation marks are used here because 
nowadays it is commonly accepted that this unconventional representation 
of observables is consistent with fundamental axioms of quantum mechanics
if, and only if, 
the non-self-adjoint representative~$P$ is \emph{quasi-self-adjoint},
\ie
\begin{equation}\label{metric}
  P^* = \Theta \, P \, \Theta^{-1}
\end{equation}
with some positive, bounded and boundedly invertible operator~$\Theta$
called \emph{metric}
(the special choice $\Theta=I$ corresponds to self-adjointness).
But then one is back in the conventional quantum mechanics just by
modifying the inner product $(\cdot,\cdot)$ in the underlying
Hilbert space to $(\cdot,\Theta\cdot)$.
The notion of quasi-self-adjoint
(then called \emph{quasi-Hermitian}) operators
in quantum mechanics was first used 
by nuclear physicists Scholtz, Geyer and Hahne in 1992~\cite{GHS}, 
but it was actually considered previously
by the mathematician Dieudonn\'e as early as in 1961~\cite{Dieudonne_1961}.
We refer to the review article~\cite{KSTV}
and the book chapter~\cite{KS-book} for mathematical aspects 
of quasi-self-adjoint quantum mechanics.  

Given a non-self-adjoint operator~$P$ with real spectrum,
it is usually not easy to decide whether it is quasi-self-adjoint.
Even if the latter is known to hold for an operator~$P$
(like for instance for operators whose eigenfunctions form a Riesz basis), 
one cannot expect to be able to solve the operator equation~\eqref{metric}
and find the metric operator~$\Theta$ in a closed form. 
Distinguished exceptions are represented by 
one-dimensional Schr\"odinger operators 
with non-self-adjoint point interactions
on intervals
\cite{KBZ,K4,Siegl_2008,Ergun_2010,Ergun-Saglam_2010,Ergun_2013,KSZ}
or graphs
\cite{HKS,Astudillo-Kurasov-Usman_2015,Kurasov-Garjani_2017}.
We particularly refer to \cite{KSZ} where the general study 
of Sturm-Liouville operators with complex Robin boundary conditions
reveals that the metric can be expressed as the sum of the identity
and a Hilbert-Schmidt operator, reflecting the Bari basis property
of the eigenfunctions.
Moreover, spectacularly simple formulae of the integral kernels
are provided in specific ($\mathcal{PT}$-symmetric) situations. 

The second motivation of this paper comes from 
the relevance of complex magnetic fields in
these non-self-adjoint representations in quantum mechanics.
Papers on the imaginary magnetic field in quantum mechanics  
exist in the literature 
(see, \eg, \cite{Hatano-Nelson_1996,Hatano-Nelson_1998}),
but the quasi-self-adjointness does not seem to have been considered.
In this paper we introduce 
\emph{momenta with complex magnetic fields}, 
as a new class of non-self-adjoint toy models,
which are probably the simplest non-trivial examples of
quasi-self-adjoint operators in infinite-dimensional Hilbert spaces whatsoever.
It turns out that the basis properties of the eigenfunctions substantially
differ from the aforementioned models with point interactions.
Because of the technical relationship with the magnetic Hardy inequalities,
we restrict to momenta on the unit circle
$\Sphere \cong (-\pi,\pi)$,
but analogous models can be considered on intervals (even unbounded) and graphs.

\begin{Theorem}\label{Thm.momenta}
Given $a \in L^2(\Sphere;\Com)$, let $P_a$ be the maximal realisation
of the operator
\begin{equation}\label{momentum}
  -i \frac{\der}{\der x} - a(x)   
  \qquad \mbox{in} \qquad
  \sii(\Sphere) \,.
\end{equation}
The operator~$P_a$ satisfies the following properties:
\begin{enumerate}
\item[\emph{(i)}]
The similarity relation
\begin{equation}\label{similar}
  \Omega_a \, P_a \, {\Omega_a}^{\!\!-1} = P_{\langle a \rangle}
\end{equation}
holds, where
\begin{equation}\label{transform}
  (\Omega_a\psi)(x) := \exp\left(
  i \, \langle a \rangle \, x
  -i \int_{-\pi}^x a(\xi) \, \der\xi  
  \right) \psi(x)  
  \qquad \mbox{with} \qquad
  \langle a \rangle := \frac{1}{2\pi} \int_{-\pi}^\pi a(x) \, \der x 
  \,.
\end{equation}
\item[\emph{(ii)}]
The operator~$P_a$ is quasi-self-adjoint if, and only if, 
\begin{equation}\label{quasi}
  \langle \Im a \rangle  = 0 
  \,.
\end{equation}
In this case, $P_a$~satisfies the relation~\eqref{metric}
with the metric given by the multiplication operator
\begin{equation}\label{metric.ours}
  (\Theta_a\psi)(x) := \exp\left(
  2 \int_{-\pi}^x \Im a(\xi) \, \der\xi 
  \right)
  \psi(x)
  \,.
\end{equation}
\item[\emph{(iii)}]
The eigenfunctions of~$P_a$ form a Riesz basis,
but not a Bari basis unless $\Im a = 0$. 
\end{enumerate}
\end{Theorem}

We emphasise that the similarity relation~\eqref{similar} holds in general,
in particular without assuming~\eqref{quasi}, 
but the transformed operator~$P_{\langle a \rangle}$
is self-adjoint if, and only if, this condition holds.
Condition~\eqref{quasi} is a complex extension 
of the requirement that the magnetic field on the circle
can be gauged out. 
Notice also (\cf~\eqref{sa})
that~$P_a$ is self-adjoint if, and only if, 
$\Im a = 0$,
while~\eqref{quasi} can be of course satisfied 
in more general situations.

\section{Momenta with complex magnetic fields}
%
In this section we restrict to the one-dimensional model~\eqref{momentum}.
We identify the unit circle~$\Sphere$ with the open interval $(-\pi,\pi)$
where the boundary points~$\pm\pi$ are glued together.

\subsection{Definition}
We introduce~$P_a$ as the operator in $\sii((-\pi,\pi))$ defined by
\begin{equation}\label{operator}
  (P_a \psi)(x) := -i\psi'(x) - a(x)\psi(x)
  \,, \qquad 
  \Dom P_a := \left\{
  \psi \in W^{1,2}((-\pi,\pi)) :
  \psi(-\pi) = \psi(\pi)
  \right\}
  \,.
\end{equation}
Our standing assumption is that the function
$a : (-\pi,\pi) \to \Com$ representing a complex-valued magnetic potential
is square-integrable, 
\ie, 
\begin{equation}\label{standing}
  a \in \sii((-\pi,\pi))
  \,.
\end{equation}
In the magnetic-free case, 
$P_0$~is the usual momentum operator in quantum mechanics,
which is well known to be self-adjoint.

Under the hypothesis~\eqref{standing}, 
the magnetic part of~$P_a$ is a small perturbation of~$P_0$.
Indeed, for every $\psi \in W^{1,2}((-\pi,\pi))$ and positive~$\eps$,
we have the bound
\begin{equation}\label{relative}
  \|a\psi\|^2 \leq \|a\|^2 \|\psi\|_\infty^2
  \leq \eps \, \|a\|^2 \|\psi'\|^2 
  + \|a\|^2 \left( \frac{1}{\eps} + \frac{1}{2\pi} \right) \|\psi\|^2
  \,,
\end{equation}
where $\|\cdot\|$ is the norm of $\sii((-\pi,\pi))$
and $\|\cdot\|_\infty$ is the supremum norm.
Here the second inequality is due to the Sobolev embedding
$W^{1,2}((-\pi,\pi)) \hookrightarrow C^0([-\pi,\pi])$
quantified by
$$
  \|\psi\|_\infty^2 \leq 2 \|\psi\| \|\psi'\| + \frac{1}{2\pi} \|\psi\|^2
$$
and an elementary Young-type inequality. 
It follows from~\eqref{relative} that~$a$,
considered as the multiplication operator in $\sii((-\pi,\pi))$,
is relatively bounded with respect to~$P_0$,
with the relative bound equal to zero.
By standard perturbation results \cite[Sec.~IV.1]{Kato}, 
$P_a$ is a well defined closed operator 
with compact resolvent.

\begin{Remark}\label{Rem.weight}
Let $\psi \in W^{1,2}((-\pi,\pi)) $
and consider the changed function 
$$
  \phi(x) := w(x) \, \psi(x)
  \qquad \mbox{with} \qquad
  w(x) := \exp\left(
  \int_{-\pi}^x \Im a(\xi) \, \der\xi
  \right) 
  \,.
$$
In view of the obvious relationship between quadratic forms
$$
  \int_{-\pi}^\pi |(P_a\psi)(x)|^2 \, \der x
  = \int_{-\pi}^\pi |(P_{\Re a}\psi)(x)|^2 \, w^{-2}(x) \, \der x
  \,,
$$
one can interpret the imaginary part of the magnetic field
as a way to handle the conventional (self-adjoint) 
magnetic Schr\"odinger operator in a weighted space.
Notice also that $\phi \in \Dom P_a$
provided that $\psi \in \Dom P_a$
and the condition~\eqref{quasi} holds.
\end{Remark}

\subsection{Obvious symmetries}
The adjoint of~$P_a$ satisfies
\begin{equation}\label{sa}
  {P_a}^{\!*} = P_{\bar{a}} \,,
\end{equation}
where~$\bar{a}$ denotes the complex conjugate of~$a$. 
Consequently, $P_a$~is self-adjoint if, and only if, 
the imaginary part of~$a$ identically equals zero, 
\ie~$\Im a = 0$.

Let the \emph{time-reversal} operator~$\mathcal{T}$
and \emph{parity} operator~$\mathcal{P}$ be defined by 
the usual involutions
$$
  (\mathcal{T}\psi)(x) := \overline{\psi(x)}
  \qquad \mbox{and} \qquad
  (\mathcal{P}\psi)(x) := \psi(-x)
  \,.
$$
Then we obviously have the relations
$$
  \mathcal{T} P_a \mathcal{T} = - P_{-\mathcal{T}a}
  \,, \qquad 
  \mathcal{P} P_a \mathcal{P} = - P_{-\mathcal{P}a}
  \qquad \mbox{and} \qquad
  (\mathcal{PT}) P_a (\mathcal{PT}) = P_{\mathcal{PT}a}
  \,.
$$
Consequently, $P_a$~is anti-$\mathcal{P}$-self-adjoint
(\ie, ${P_a}^{\!*} = -\mathcal{P} P_a \mathcal{P}$) if, and only if,
$\Re a$ is odd and $\Im a$ is even;
and $P_a$~is $\mathcal{PT}$-symmetric
(\ie, $[P_a,\mathcal{PT}]=0$) if, and only if,
$\Re a$ is even and $\Im a$ is odd.

The $\mathcal{PT}$-symmetry of~$P_a$
is a sufficient (but not necessary) condition
to guarantee the quasi-self-adjointness condition~\eqref{quasi}.

\begin{Remark}
Instead of the momentum operator~$P_a$,
one can also consider the non-self-adjoint magnetic Hamiltonian 
$H_a := {P_a}^{\!2}$.
Then it follows that~$H_a$ is $\mathcal{P}$-self-adjoint
(\ie, ${H_a}^{\!*} = \mathcal{P} H_a \mathcal{P}$) if, and only if,
$\Re a$ is odd and $\Im a$ is even;
and $H_a$~is $\mathcal{PT}$-symmetric if, and only if,
$\Re a$ is even and $\Im a$ is odd.
This is an example of operator where these two notions 
are not compatible unless $a=0$.
At the same time, $H_a$ is not $\mathcal{T}$-self-adjoint unless $a=0$.
\end{Remark}

\subsection{Spectrum}
Since~$P_a$ has a compact resolvent, its spectrum is purely discrete,
\ie~composed of isolated eigenvalues with finite algebraic multiplicities.
In this one-dimensional setting, finding the eigenvalues is a routine:
solving the differential equation $-i\psi'-a(x)\psi=\lambda\psi$
in terms of exponential functions and subjecting the solutions 
to the periodic boundary conditions of~\eqref{operator},
one immediately arrives at
\begin{equation}\label{spectrum}
  \sigma(P_a) = \{ m - \langle a\rangle \}_{m \in \Int}
  \,.
\end{equation}

Notice that the spectrum is real if, and only if, 
the quasi-self-adjointness condition~\eqref{quasi} holds.
Again, the $\mathcal{PT}$-symmetry of~$P_a$ 
is a sufficient (but not necessary) condition
to guarantee that the spectrum is real.
Irrespectively of whether~$P_a$ is $\mathcal{PT}$-symmetric or not, 
the spectrum of the adjoint~${P_a}^{\!*}$ is obtained  
by mirroring the eigenvalues of~$P_a$ with respect to the real axis.

By the procedure described above, one also finds 
the eigenfunctions of~$P_a$ 
corresponding to the eigenvalues in~\eqref{spectrum},
$$
  \psi_m(x) := \frac{1}{\sqrt{2\pi}} \,
  \exp\left(
  i \, (m-\langle a \rangle) \, x 
  + i \int_{-\pi}^x a(\xi) \, \der\xi 
  \right)
  \,.
$$
The eigenfunctions of the adjoint~${P_a}^{\!*}$
are given by
$$
  \phi_m(x) := \frac{1}{\sqrt{2\pi}} \,
  \exp\left(
  i \, (m-\langle \bar{a} \rangle) \, x 
  + i \int_{-\pi}^x \bar{a}(\xi) \, \der\xi 
  \right)
  \,.
$$
The normalisation factors are chosen in such a way that
the standard biorthogonal condition

\begin{equation}\label{bi}
  (\phi_n,\psi_m) = \delta_{nm}
\end{equation}
holds for every $m,n \in \Int$,
where $(\cdot,\cdot)$ denotes the inner product of $\sii((-\pi,\pi))$.
This condition reduces to 
$
  (\psi_n,\psi_m) = \delta_{nm}
$
if $\Im a =0$ (so that~$P_a$ is self-adjoint).

The availability of the condition~\eqref{bi} ensures
that the spectrum of~$P_a$ is semisimple 
(\ie~the algebraic and geometric multiplicities 
of the eigenvalues are equal).
From~\eqref{spectrum} we then immediately conclude
that the spectrum is actually simple 
(\ie~all the eigenvalues have multiplicity one).

\subsection{Basis properties}
For various basis properties of a sequence in a Hilbert space,
we refer to \cite{Gohberg-Krein_1969} and \cite[Sec.~3]{Davies_2007}.

The principal observation in our case is that,
for every $m \in \Int$,
\begin{equation}\label{crucial}
  \psi_m(x) = \xi(x) \, e_m(x) 
  \qquad \mbox{and} \qquad
  \phi_m(x) = \xi^{-1}(x) \, e_m(x) 
  \,,
\end{equation}
where
$$
  \xi(x) :=
  \exp\left(  
  \langle \Im a \rangle \, x - \int_{-\pi}^x \Im a(\xi) \, \der\xi
  \right)
$$
is a bounded and positive function on $[-\pi,\pi]$,
while $\{e_m\}_{m \in \Int}$ is an orthonormal basis in $\sii((-\pi,\pi))$
(for it corresponds to the choice $\Im a = 0$ when $P_a$ is self-adjoint).
Consequently, it immediately follows that the sequence 
$\{\psi_m\}_{m \in \Int}$ is \emph{almost normalised} in $\sii((-\pi,\pi))$
in the sense that 
$$
  \inf_{m \in \Int}\|\psi_m\| > 0
  \qquad \mbox{and} \qquad
  \sup_{m \in \Int}\|\psi_m\| < \infty 
  \,.
$$

Next, assuming $(\psi_m,\psi) = 0$ for every $m \in \Int$
and using the completeness of $\{e_m\}_{m \in \Int}$,
we conclude that $\xi \psi = 0$, and therefore $\psi = 0$.
This argument shows that $\{\psi_m\}_{m \in \Int}$
is a \emph{complete} set in $\sii((-\pi,\pi))$.
In view of~\eqref{bi}, we also know that $\{\psi_m\}_{m \in \Int}$
is \emph{minimal complete} (\ie~removal of any element makes it incomplete).

Similarly, since $\{e_m\}_{m \in \Int}$ is in particular 
a (\emph{Schauder} or \emph{conditional}) \emph{basis}, 
given any function $\psi \in \sii((-\pi,\pi))$,
we have the unique decomposition
$$
  \psi 
  = \sum_{m \in \Int} (e_m,\psi) \, e_m
  = \sum_{m \in \Int} (\xi^{-1} e_m,\psi) \, \xi e_m
  \,.
$$
Using~\eqref{crucial},
we conclude that $\{\psi_m\}_{m \in \Int}$
is a basis in $\sii((-\pi,\pi))$.

In fact, having the completeness of $\{\psi_m\}_{m \in \Int}$,
we immediately conclude from~\eqref{crucial} that
$\{\psi_m\}_{m \in \Int}$ is a \emph{Riesz} (or \emph{unconditional}) \emph{basis} 
in $\sii((-\pi,\pi))$,
for it is obviously equivalent to an orthonormal basis
(\cf~\cite[Thm.~3.4.5]{Davies_2007}).
This shows one part of the claim~(iii) of Theorem~\ref{Thm.momenta}.
Because of the biorthogonality relation~\eqref{bi},
we also know that $\{\phi_m\}_{m \in \Int}$ 
is a Riesz basis in $\sii((-\pi,\pi))$.

It remains to prove that $\{\psi_m\}_{m \in \Int}$
is not a \emph{Bari basis} if $\Im a \not= 0$, 
meaning that it is not quadratically close to any orthonormal basis.
By \cite[Sec.~VI.3]{Gohberg-Krein_1969},
it is enough to show that the biorthogonal bases
$\{\psi_m\}_{m \in \Int}$ and $\{\phi_m\}_{m \in \Int}$
are not quadratically close to each other.
Using that $|e_m| = 1/\sqrt{2\pi}$ for every $m \in \Int$
and~\eqref{crucial}, we arrive at
$|\psi_m-\phi_m| = |\xi-\xi^{-1}|/\sqrt{2\pi}$
which is independent of~$m$.
Consequently,
$$
  \sum_{m \in \Int} \|\psi_m-\phi_m\|^2 = \infty
$$
unless $\xi=\xi^{-1}$, 
which is possible only if $\Im a = 0$.
This concludes the proof of the claim~(iii) of Theorem~\ref{Thm.momenta}.

\subsection{Hidden symmetries}
Now we turn to the proof of the remaining claims 
of Theorem~\ref{Thm.momenta}.

First of all, notice that~$\Omega_a$ 
from the claim~(i) of Theorem~\ref{Thm.momenta}
is a multiplication operator
generated by a bounded and non-zero differentiable function on~$[-\pi,\pi]$.
A formal verification of 
the similarity relation~\eqref{similar} is a routine.
To make it rigorous, it remains to show that
$\Omega_a$ leaves the domain of~$P_a$
(which is independent of~$a$, see~\eqref{operator}) invariant. 
The principal step is to verify the invariance 
of the periodic boundary conditions.
Given any $\psi \in W^{1,2}((-\pi,\pi))$, 
it is straightforward to check (recall also~\eqref{relative})
that $\Omega_a\psi \in W^{1,2}((-\pi,\pi))$ and 
$$
  (\Omega_a\psi)(-\pi) 
  = e^{-i \langle a \rangle \pi} \psi(-\pi) 
  \,, \qquad
  (\Omega_a\psi)(\pi) 
  = e^{i \langle a \rangle \pi} e^{-i 2\pi \langle a \rangle } \psi(-\pi)
  \,.
$$
Hence, if~$\psi$ additionally satisfies $\psi(\pi) = \psi(-\pi)$,
then it is also the case for $\Omega_a\psi$.
Summing up, we have verified the claim~(i) of Theorem~\ref{Thm.momenta}.

Now we turn to the claim~(ii) of Theorem~\ref{Thm.momenta}.
The operator~$P_{\langle a \rangle}$ is self-adjoint
if, and only if, the condition~\eqref{quasi} holds.
In this case, $P_a$ is similar to a self-adjoint operator
in view of~\eqref{similar}
and it clearly satisfies the relation~\eqref{metric}
with the metric operator 
$
  \Theta_a = {\Omega_a}^{\!\!*} \, \Omega_a
$,
which coincides with formula~\eqref{metric.ours}.
Conversely, if~$P_a$ satisfies the relation~\eqref{metric}
with some positive, bounded and boundedly invertible operator~$\Theta$,
then it is clearly similar to the self-adjoint operator 
$\Theta^{1/2} P_a \Theta^{-1/2}$.
Hence, the spectrum of~$P_a$ is necessarily real,
which is equivalent to the condition~\eqref{quasi}
due to the explicit formula~\eqref{spectrum}.
So we are back in the situation 
where~$P_{\langle a \rangle}$ is self-adjoint.
This concludes the proof of Theorem~\ref{Thm.momenta}.

\begin{Remark}
Irrespectively of whether the condition~\eqref{quasi} holds or not,
$P_a$~is always similar to the \emph{normal} operator~$P_{\langle a \rangle}$.
\end{Remark}
%

\section{Hardy inequalities with complex magnetic fields}
%
This section is devoted to proofs of 
Theorems~\ref{Thm.Hardy} and~\ref{Thm.AB} from the introduction
and other related results.
The standing assumption for regular fields 
is that the function $B: \Real^2 \to \Com$ representing 
a complex magnetic field is (infinitely) smooth, 
\ie\
$
  B \in C^\infty(\Real^2)
$.
For the proof of Theorem~\ref{Thm.Hardy} 
we shall additionally assume that~$B$ has a compact support,
\ie\
$
  B \in C_0^\infty(\Real^2)
$.
It will be clear from the proof below
that much less is actually needed to establish Theorem~\ref{Thm.Hardy}.
Since~$B$ is complex-valued, the freedom of the gauge invariance is lost
and we stick (except for Section~\ref{Sec.robust}) to the choice of 
the corresponding vector potential $A:\Real^2 \to \Com^2$
given by~\eqref{gauge}.
Because of the smoothness assumption about the field~$B$,
we then also have
$
  A \in C^\infty(\Real^2;\Com^2)
$.

\subsection{Passing to polar coordinates}
The main idea coming back to~\cite{Laptev-Weidl_1999}
(see also~\cite{K7} and~\cite{CK})
is to employ polar coordinates 
$x=(x_1,x_2) = (r\cos\theta,r\sin\theta)$
with $x \in \Real^2$ 
and $r \in (0,\infty)$, $\theta \in (-\pi,\pi)$.
Then~\eqref{gauge} acquires a more transparent form
\begin{equation}\label{passing}
  A(x) = (-\sin\theta,\cos\theta) \,
  \frac{a(r,\theta)}{r}
  \,, \qquad \mbox{where} \qquad
  a(r,\theta) := \int_0^r B(t \cos\theta,t\sin\theta) \, t \, \der t
  \,.
\end{equation}
In particular, we have
$
  a \in C^\infty((0,\infty)\times[-\pi,\pi])
$.

Recall the mean-value notation $\langle \cdot \rangle$
introduced in~\eqref{transform}.
If~$B$ has a compact support, 
then there clearly exists a positive number~$R$ such that
$
  \langle a(r,\cdot) \rangle = \langle a(R,\cdot) \rangle
$ 
for all $r \geq R$.
The flux condition~\eqref{flux} is then equivalent to
\begin{equation}\label{flux.intro.bis}
  \lim_{r \to \infty}
  \big\langle \Re a(r,\cdot) \big\rangle \not\in \Int
  \qquad \mbox{or} \qquad
  \lim_{r \to \infty}
  \big\langle \Im a(r,\cdot) \big\rangle \not=0
  \,.
\end{equation}
Occasionally, we shall also consider the weaker condition
\begin{equation}\label{flux.bis}
  \exists r \in (0,\infty), \qquad
  \big\langle \Re a(r,\cdot) \big\rangle \not\in \Int
  \qquad \mbox{or} \qquad
  \big\langle \Im a(r,\cdot) \big\rangle \not=0
  \,,
\end{equation}
for which we do not need to assume that~$B$ is compactly supported.	

Let $\psi \in C_0^\infty(\Real^2)$.
Writing $\phi(r,\theta) := \psi(r\cos\theta,r\sin\theta)$
and using the transverse property~\eqref{transverse} of~$A$,
we have
\begin{equation}\label{polar}
  \int_{\Real^2} |\nabla_{\!A}\psi(x)|^2 \, \der x
  = \int_0^\infty \int_{-\pi}^\pi 
  \left(
  |\partial_r\phi(r,\theta)|^2 
  + \frac{|-i\partial_\theta\phi(r,\theta)-a(r,\theta)\phi(r,\theta)|^2}{r^2}
  \right)
  \der\theta \, r \, \der r 
  \,.
\end{equation}
Recalling~\eqref{operator},
it is clear that the concept of momenta on the circle 
with complex magnetic fields developed in the precedent section
will play an important role in the sequel.

\subsection{Local Hardy inequalities}
For every positive~$r$, let us define a non-negative number
\begin{equation}\label{lambda}
  \lambda_a(r) := \inf_{\stackrel[\varphi\not=0]{}{
  \varphi \in \Dom P_a}}
  \frac{\displaystyle
  \int_{-\pi}^\pi 
  |(P_{a(r,\cdot)}\varphi)(\theta)|^2 \, \der\theta}
  {\displaystyle
  \int_{-\pi}^\pi 
  |\varphi(\theta)|^2 \, \der\theta}
  \,,
\end{equation}
where $\Dom P_a$ is introduced in~\eqref{operator}.
Clearly, $\lambda_a(r)$ is the lowest point in the spectrum
of the self-adjoint operator 
${P_{a(r,\cdot)}}^{\!\!*} \, P_{a(r,\cdot)}$
in $\sii((-\pi,\pi))$,
where~$r$ is regarded as a parameter.
If~$B$ has a compact support, 
then there clearly exists a positive number~$R$ such that
$
  \lambda_a(r) = \lambda_a(R)
$ 
for all $r \geq R$.

Using the definition~\eqref{lambda} in~\eqref{polar}
together with Fubini's theorem,
we arrive at the following type of Hardy inequality
\begin{equation}\label{local}
\begin{aligned}
  \int_{\Real^2} |\nabla_{\!A}\psi(x)|^2 \, \der x
  &\geq \int_0^\infty \int_{-\pi}^\pi 
  \left(
  |\partial_r\phi(r,\theta)|^2 
  + \frac{\lambda_a(r)}{r^2} \, |\phi(r,\theta)|^2 
  \right)
  \der\theta \, r \, \der r 
  \\
  &\geq \int_0^\infty \int_{-\pi}^\pi 
  \frac{\lambda_a(r)}{r^2} \, |\phi(r,\theta)|^2 
  \, \der\theta \, r \, \der r 
  \\
  &= \int_{\Real^2} \frac{\lambda_a(|x|)}{|x|^2} \, |\psi(x)|^2 \, \der x
  \,.
\end{aligned}
\end{equation}
We remark (\cf~\eqref{passing}) that 
$\lambda_a(r) = \mathcal{O}(r^2)$ as $r \to 0$,
so that $\lambda_a(r)/r^2$ has actually no singularity at $r=0$.

We call the result~\eqref{local} a \emph{local} Hardy inequality
because~$\lambda_a$ may not be everywhere positive.  
However, the following lemma ensures that~$\lambda_a$
is a non-trivial non-zero function on $(0,\infty)$
whenever the condition~\eqref{flux.bis} holds.

\begin{Lemma}\label{Lem.lambda}
The function $r \mapsto \lambda_a(r)$ is continuous.
Given any $r \in (0,\infty)$,  
$\lambda_a(r) = 0$ if, and only if,
\begin{equation}\label{antiflux}
  \big\langle \Re a(r,\cdot) \big\rangle \in \Int
  \qquad \mbox{and} \qquad
  \big\langle \Im a(r,\cdot) \big\rangle =0
  \,.
\end{equation}
\end{Lemma}
\begin{proof}
Since $\Dom P_a \subset W^{1,2}((-\pi,\pi))$ 
is compactly embedded in $\sii((-\pi,\pi))$,
the spectrum of the operator 
${P_{a(r,\cdot)}}^{\!\!*} \, P_{a(r,\cdot)}$
is purely discrete.
Consequently, the infimum in~\eqref{lambda} is achieved
and $\lambda_a(r)$ is the lowest eigenvalue of this operator.
Let us denote by~$\varphi_r$ a corresponding eigenfunction.
 
First of all, we notice that,
using a constant as a test function in~\eqref{lambda},
we have the upper bound
\begin{equation}\label{bound}
  \lambda_a(r) \leq \langle |a(r,\cdot)|^2 \rangle 
  \,.
\end{equation}
Using~$\varphi_{r_2}$ as a test function for $\lambda_a(r_1)$,
it follows from~\eqref{lambda} that
$$
  \lambda_a(r_1) 
  \leq 
  \frac{\displaystyle
  \int_{-\pi}^\pi 
  \big|(P_{a(r_2,\cdot)}\varphi_{r_2})(\theta)
  + [a(r_2,\theta)-a(r_1,\theta)]\varphi_{r_2}(\theta)\big|^2 
  \, \der\theta}
  {\displaystyle
  \int_{-\pi}^\pi 
  |\varphi_{r_2}(\theta)|^2 \, \der\theta}
  \leq 
  \left(
  \sqrt{\lambda_a(r_2)} 
  + \sigma(r_1,r_2)
  \right)^2
$$
with
$$
  \sigma(r_1,r_2) := 
  \sup_{\theta\in(-\pi,\pi)} |a(r_2,\theta)-a(r_1,\theta)|
  \,.
$$
Using additionally the analogous estimate obtained by
reversing the role of~$r_1$ and~$r_2$ together with~\eqref{bound}, 
one eventually gets
$$
  |\lambda_a(r_2) - \lambda_a(r_1)|
  \leq \sigma(r_1,r_2)^2
  + 2 \, \sigma(r_1,r_2) \,
  \max\left\{ 
  \sqrt{\langle |a(r_1,\cdot)|^2},\sqrt{\langle |a(r_2,\cdot)|^2} 
  \right\}
  \,.
$$
By virtue of the smoothness of~$a$,
this inequality ensures that~$\lambda_a$ is continuous
(even Lipschitz).

Let us now fix any positive~$r$
and show that~\eqref{antiflux}
implies that $\lambda_a(r)=0$. 
Using the similarity relation~\eqref{similar},
one has the estimate
\begin{equation}\label{crude}
  \lambda_a(r) \leq 
  \kappa_a(r)^2 \
  \inf_{\stackrel[\varphi\not=0]{}{
  \varphi \in \Dom P_a}}
  \frac{\displaystyle
  \int_{-\pi}^\pi 
  |(P_{\langle a(r,\cdot) \rangle}
  \Omega_{a(r,\cdot)}\varphi)(\theta)|^2 \, \der\theta}
  {\displaystyle
  \int_{-\pi}^\pi 
  |(\Omega_{a(r,\cdot)}\varphi)(\theta)|^2 \, \der\theta}
  \,,
\end{equation}
where the condition number
$\kappa_a(r) := \|\Omega_{a(r,\cdot)}\| \|{\Omega_{a(r,\cdot)}}^{\!\!-1}\|$
satisfies $1 \leq \kappa_a(r) < \infty$.
Here $\|\cdot\|$ stands for the operator norm in $\sii((-\pi,\pi))$,
which can be expressed through a supremum norm in our case
for the operator in question is a multiplication operator.
In view of the second condition of~\eqref{antiflux},
the operator $P_{\langle a(r,\cdot) \rangle}$ is self-adjoint,
and therefore the Rayleigh quotient in~\eqref{crude} 
equals the square of its smallest eigenvalue in absolute value.
Recalling~\eqref{spectrum}, we thus obtain
\begin{equation}\label{crude1}
  \lambda_a(r) \leq 
  \kappa_a(r)^2 \, 
  \dist(\langle \Re a(r,\cdot) \rangle,\Int)^2
  \,.
\end{equation}
Using now in addition the first condition of~\eqref{antiflux},
we conclude that $\lambda_a(r)=0$.

Conversely, assuming that $\lambda_a(r)=0$ for a given positive~$r$,
let us show that necessarily~\eqref{antiflux} holds.
Indeed, if $\lambda_a(r)=0$, it follows from~\eqref{lambda}
that $P_{a(r,\cdot)}\varphi_r=0$, 
whence $0 \in \sigma(P_{a(r,\cdot)})$.
Recalling~\eqref{spectrum}, 
the latter is possible only if~\eqref{antiflux} holds.
\end{proof}
\begin{Remark}
Proceeding as in~\eqref{crude} and~\eqref{crude1},
and assuming the second condition of~\eqref{antiflux},
one also has the lower bound
$$
  \lambda_a(r) \geq 
  \kappa_a(r)^{-2} \, 
  \dist(\langle \Re a(r,\cdot) \rangle,\Int)^2
  \,.
$$
\end{Remark}

\subsection{From local to global Hardy inequalities}
Now we explain how to pass from 
the local Hardy inequality~\eqref{local} to a \emph{global} one
(\ie~with a Hardy weight positive everywhere in~$\Real^2$). 
Notice that the following theorem does not require 
compactly supported magnetic fields
and the weaker hypothesis~\eqref{flux.bis} is assumed.

\begin{Theorem}\label{Thm.Hardy.log}
Let a smooth complex field $B:\Real^2\to\Com$ 
satisfy the condition 
\begin{equation}\label{flux.log}
  \exists r \in (0,\infty), \qquad
  \frac{1}{2\pi}\int_{D_r} \Re B(x) \, \der x \not\in \Int
  \qquad \mbox{or} \qquad
  \frac{1}{2\pi}\int_{D_r} \Im B(x) \, \der x \not=0
  \,,
\end{equation}
where $D_r := \{x \in \Real^2 : |x| < r\}$ is the disk of radius~$r$.
Then there exists a positive constant~$\tilde{c}$ 
depending on~$B$
such that the following inequality holds
\begin{equation}\label{Hardy.log}
  \forall \psi \in C_0^\infty(\Real^2) \,, \qquad
  \int_{\Real^2} |\nabla_{\!A}\psi(x)|^2 \, \der x
  \geq \tilde{c} \int_{\Real^2} \frac{|\psi(x)|^2}{1+|x|^2\log^2(|x|)} \, \der x
  \,,
\end{equation}
where the complex vector potential~$A$ is given by~\eqref{gauge}.   
\end{Theorem}
\begin{proof}
The claim follows by mimicking the proof of \cite[Thm.~3.1]{CK}.
Let $\psi \in C_0^\infty(\Real^2)$ and recall that
we denote by~$\phi$ its counterpart in polar coordinates. 

From Lemma~\ref{Lem.lambda}
it follows that there exists a positive constant~$\nu$
(depending on the behaviour of the function~$\lambda_a$,
which is in turn determined by the behaviour of~$B$)
and a bounded open interval $I \subset (0,\infty)$ such that
$\lambda_a(r) / r^2 \geq \nu > 0$ for all $r \in I$.
From~\eqref{local} we thus conclude
\begin{equation}\label{step1}
  \int_{\Real^2} |\nabla_{\!A}\psi(x)|^2 \, \der x
  \geq
  \nu \int_{\Real^2} \chi_I(x) \, |\psi(x)|^2 \, \der x
  \,,
\end{equation}
where~$\chi_I$ denotes the characteristic function
of the annulus $\{x \in \Real^2 : |x| \in I\}$.

To extend this local Hardy inequality to~$\Real^2$,
we employ the presence of the kinetic term that
we neglected in~\eqref{local}, namely
\begin{equation}\label{step1.5}
  \int_{\Real^2} |\nabla_{\!A}\psi(x)|^2 \, \der x
  \geq
  \int_0^\infty \int_{-\pi}^\pi 
  |\partial_r\phi(r,\theta)|^2 
  \, \der\theta \, r \, \der r 
  \,,
\end{equation}
together with the one-dimensional Hardy-type inequalities
(\cf~\cite[Lem.~3.1]{CK})
\begin{equation}\label{aux2}
  \forall f \in C_0^\infty(\Real\setminus\{r_0\}) 
  \,, \qquad
\begin{aligned}  
  \int_{0}^{r_0}
  |f'(r)|^2 \, r \, \der r
  &\geq \gamma
  \int_{0}^{r_0} |f(r)|^2 \, r \, \der r
  \,,
  \\
  \int_{r_0}^\infty
  |f'(r)|^2 \, r \, \der r
  &\geq \gamma
  \int_{r_0}^\infty \frac{|f(r)|^2}{r^2\log^2(r/r_0)} \, r \, \der r
  \,,
\end{aligned}
\end{equation}
valid for any positive~$r_0$
with some positive constant~$\gamma$
depending only on~$r_0$. 
Choosing~$r_0$ to be the middle point of the interval~$I$,
we introduce a cut-off function
$\xi \in C^\infty((0,\infty))$
such that $0 \leq \xi \leq 1$,
$\xi$~vanishes in a neighbourhood of~$r_0$
and $\xi=1$ outside the interval~$I$.
Writing 
$
  \phi(r,\theta) = \xi(r) \phi(r,\theta) + (1-\xi(r)) \phi(r,\theta)
$
and using~\eqref{aux2} with help of Fubini's theorem, 
we get
\begin{eqnarray}\label{trick}
\lefteqn{
  \int_{-\pi}^\pi \int_0^\infty  
  \frac{|\phi(r,\theta)|^2}{1+r^2\log^2(r/r_0)} 
  \, r \, \der r \, \der \theta
}
  \\
  &&\leq
  2 \int_{-\pi}^\pi \int_0^\infty 
  \frac{|\xi(r) \phi(r,\theta)|^2}{1+r^2\log^2(r/r_0)} 
  \, r \, \der r \, \der \theta
  +  2 \int_{-\pi}^\pi \int_0^\infty 
  |(1-\xi(r)) \phi(r,\theta)|^2 
  \, r \, \der r \, \der \theta
  \nonumber \\
  &&\leq
  \frac{2}{\gamma} \int_{-\pi}^\pi \int_0^\infty 
  |\partial_r[\xi(r) \phi(r,\theta)]|^2 
  \, r \, \der r \, \der \theta
  +  2 \int_{-\pi}^\pi \int_I 
  |\phi(r,\theta)|^2 
  \, r \, \der r \, \der \theta
  \nonumber \\
  &&\leq
  \frac{4}{\gamma} \int_{-\pi}^\pi \int_0^\infty 
  |\partial_r\phi(r,\theta)|^2 
  \, r \, \der r \, \der \theta
  + \frac{4}{\gamma} \int_{-\pi}^\pi \int_I 
  |\xi'(r)|^2 \, |\phi(r,\theta)|^2 
  \, r \, \der r \, \der \theta
  +  2 \int_{-\pi}^\pi \int_I 
  |\phi(r,\theta)|^2 
  \, r \, \der r \, \der \theta
  \nonumber \\
  &&\leq
  \frac{4}{\gamma} \int_{-\pi}^\pi \int_0^\infty 
  |\partial_r\phi(r,\theta)|^2 
  \, r \, \der r \, \der \theta
  + \left( \frac{4}{\gamma} \, \|\xi'\|_\infty^2 + 2 \right)
  \int_{-\pi}^\pi \int_I 
  |\phi(r,\theta)|^2 
  \, r \, \der r \, \der \theta
  \,.
  \nonumber
\end{eqnarray}
Coming back to the test function~$\psi$
and recalling~\eqref{step1.5},
we have therefore proved
\begin{equation}\label{step2}
  \int_{\Real^2} |\nabla_{\!A}\psi(x)|^2 \, \der x
  \geq \frac{\gamma}{4} \int_{\Real^2}
  \frac{|\psi(x)|^2}{1+|x|^2\log^2(|x|/r_0)} \, \der x
  - \left(\|\xi\|_\infty^2+\frac{\gamma}{2}\right)
  \int_{\Real^2}
  \chi_I(x) \, |\psi(x)|^2 \, \der x
  \,.
\end{equation}

Finally, interpolating between~\eqref{step1} and~\eqref{step2},
we get
$$
  \int_{\Real^2} |\nabla_{\!A}\psi(x)|^2 \, \der x
  \geq
  \left[
  (1-\eps) \nu - \eps \left(\|\xi\|_\infty^2+\frac{\gamma}{2}\right)
  \right]
  \int_{\Real^2} \chi_I(x) \, |\psi(x)|^2 \, \der x
  + \eps \, \frac{\gamma}{4} \int_{\Real^2}
  \frac{|\psi(x)|^2}{1+|x|^2\log^2(|x|/r_0)} \, \der x
$$
with any $\eps > 0$.
Choosing~$\eps$ in such a way that the square bracket vanishes,
we obtain~\eqref{Hardy.log} with
$$
  \tilde{c} \geq
  \frac{\displaystyle \frac{\gamma}{4} \, \nu}
  {\displaystyle \nu+\|\xi\|_\infty^2+\frac{\gamma}{2}} \
  \inf_{r \in (0,\infty)} \frac{1+r^2\log^2(r)}{1+r^2\log^2(r/r_0)}
  >0
  \,.
$$
The theorem is proved.
\end{proof}

Now we are in a position to prove Theorem~\ref{Thm.Hardy}.
It follows as a consequence of~\eqref{local} 
and Theorem~\ref{Thm.Hardy.log} 
under the stronger hypothesis~\eqref{flux.intro.bis}.

\begin{proof}[Proof of Theorem~\ref{Thm.Hardy}]
The claim follows by mimicking the proof of \cite[Thm.~3.2]{CK}.

In view of the hypothesis~\eqref{flux} 
and since~$B$ is supposed to have a compact support,
there exists a positive number~$R$ such that~\eqref{flux.log}
holds for all $r \geq R$.  
Indeed, it is enough to choose~$R$ so large 
that the support of~$B$ is contained in the disk~$D_R$. 
As a consequence of Theorem~\ref{Thm.Hardy.log}, 
we therefore obtain
\begin{equation}\label{int}
  \int_{\Real^2} |\nabla_{\!A}\psi(x)|^2 \, \der x
  \geq
  \tilde{c} \int_{D_R} \frac{|\psi(x)|^2}{1+|x|^2\log^2(|x|)} \, \der x
  \geq
  \tilde{c} \, a_R \int_{D_R} \frac{|\psi(x)|^2}{1+|x|^2} \, \der x
  \,,
\end{equation}
where~$\tilde{c}$ is positive and
$$
  a_R := \inf_{r\in(0,R)} \frac{1+r^2}{1+r^2\log^2(r)}
$$
is also a positive constant.

At the same time, 
recalling that the function $r \mapsto \lambda_a(r)$ 
is constant on $[R,\infty)$, 
the local Hardy inequality~\eqref{local} yields
\begin{equation}\label{ext}
  \int_{\Real^2} |\nabla_{\!A}\psi(x)|^2 \, \der x
  \geq
  \lambda_a(R)
  \int_{\Real^d \setminus D_R} \frac{|\psi(x)|^2}{|x|^2} \, \der x
  \geq
  \lambda_a(R)
  \int_{\Real^d \setminus D_R} \frac{|\psi(x)|^2}{1+|x|^2} \, \der x
  \,,
\end{equation}
where $\lambda_a(R)$ is positive due to Lemma~\ref{Lem.lambda}.

Combining the two inequalities~\eqref{int} and~\eqref{ext}, 
we get~\eqref{Hardy} with
$
  c \geq \min\{\tilde{c} \, a_R,\lambda_a(R)\}
  > 0
$.
\end{proof}
\begin{Remark}\label{Rem.optimal}
The Hardy inequality of Theorem~\ref{Thm.Hardy} is optimal
in the sense that if the hypothesis~\eqref{flux} does not hold,
then~\eqref{Hardy} cannot hold with a positive constant~$c$,
namely
$$
  c_* :=
  \inf_{\stackrel[\psi\not=0]{}{
  \psi \in C_0^\infty(\Real^2)}}
  \frac{\displaystyle
  \int_{\Real^2} |\nabla_{\!A}\psi(x)|^2 \, \der x}
  {\displaystyle
  \int_{\Real^2} \frac{|\psi(x)|^2}{1+|x|^2} \, \der x} 
  = 0
  \,.
$$ 
To see it, let us pass to the polar coordinates as above.
We set $\phi(r,\theta) := f(r) \varphi_R(\theta)$,
where~$R$ is such that $D_R \supset \supp B$, 
$f$~is such that $\supp f \subset (R,\infty)$
and~$\varphi_R$ is the eigenfunction of the operator
${P_{a(R,\cdot)}}^{\!\!*} \, P_{a(R,\cdot)}$
corresponding to the eigenvalue $\lambda_a(R)$, \cf~\eqref{lambda}.
If~\eqref{flux} is violated, 
then $\lambda_a(R)=0$ due to Lemma~\ref{Lem.lambda}.
Consequently,
$$
  c_* \leq
  \inf_{\stackrel[f\not=0]{}{
  f \in C_0^\infty((R,\infty))}}
  \frac{\displaystyle
  \int_0^\infty |f'(r)|^2 \, r \, \der r}
  {\displaystyle
  \int_0^\infty \frac{|f(r)|^2}{1+|r|^2} \, r \, \der r} 
  \,.
$$ 
It is well known that this infimum equals zero.
An explicit minimising sequence 
is for example given by a mollification of 
$$
  f_n(r) := \frac{1}{\log n}
  \begin{cases}
    \log(r/n)
    & \mbox{if} \quad r \in [n,n^2] \,,
    \\
    \log(n^3/r) 
    & \mbox{if} \quad r \in [n^2,n^3] \,,
    \\
    0
    & \mbox{otherwise} \,. 
  \end{cases}
$$  
with $n > R$.
\end{Remark}

\subsection{The Aharonov-Bohm field}
Finally, we establish Theorem~\ref{Thm.AB}
dealing with the vector potential~\eqref{AB}.

\begin{proof}[Proof of Theorem~\ref{Thm.AB}]
Let $\psi \in C_0^\infty(\Real^2\setminus\{0\})$.
Passing to the polar coordinates as in~\eqref{polar}
and proceeding as in~\eqref{local},
we have 
\begin{equation*}
\begin{aligned}
  \int_{\Real^2} |\nabla_{\!A_\infty}\psi(x)|^2 \, \der x
  &= \int_0^\infty \int_{-\pi}^\pi 
  \left(
  |\partial_r\phi(r,\theta)|^2 
  + \frac{|-i\partial_\theta\phi(r,\theta)-\alpha\,\phi(r,\theta)|^2}{r^2}
  \right)
  \der\theta \, r \, \der r 
  \\
  &\geq 
  \int_0^\infty \int_{-\pi}^\pi 
  \frac{\lambda_\alpha}{r^2} \, |\phi(r,\theta)|^2 
  \, \der\theta \, r \, \der r 
  \\
  &= 
  \lambda_\alpha
  \int_{\Real^2} \frac{|\psi(x)|^2}{|x|^2} \, \der x
  \,.
\end{aligned}
\end{equation*}
The only (but significant) difference with respect 
to smooth non-trivial fields is that~$\alpha$ is a constant
(so in particular independent of the radial coordinate~$r$),
and therefore $\lambda_\alpha$ is constant.
By virtue of Lemma~\ref{Lem.lambda},
$\lambda_\alpha$ is positive if, and only if, 
the flux condition \eqref{flux.AB} holds. 
This concludes the proof of Theorem~\ref{Thm.AB}
with $c_\infty \geq \lambda_\alpha$.
\end{proof}

\subsection{Robust Hardy inequalities}\label{Sec.robust}
In this subsection we go beyond the special
choice~\eqref{gauge} of the vector potential.
Our aim is to establish a Hardy inequality 
valid for \emph{any} smooth choice of the vector potential~$A$
whenever the complex-valued magnetic field~$B$ is non-trivial
(and the imaginary part satisfies some extra condition).
We are inspired by the gauge-free proof of~\cite[Thm.~1.1]{CK},
but important modifications are needed  
because the diamagnetic inequality
is not available for complex-valued magnetic fields. 

Given a smooth complex field $B:\Real^2\to\Com$,
let $A:\Real^2 \to \Com^2$ be any smooth vector potential
satisfying $\rot A = B$.
For every positive number~$R$,
we introduce the non-negative quantity
\begin{equation}\label{mu}
  \mu_A(R) :=
  \inf_{\stackrel[\psi\not=0]{}{
  \psi \in C^\infty(\overline{D_R})}}
  \frac{\displaystyle
  \int_{D_R} |\nabla_{\!A}\psi(x)|^2 \, \der x}
  {\displaystyle
  \int_{D_R} |\psi(x)|^2 \, \der x} 
  \,.
\end{equation}
The following lemma ensures that the function $R \mapsto \mu_A(R)$ 
is identically equal to zero if, and only if, 
the magnetic field~$B$ is trivial in~$\Real^2$.

\begin{Lemma}\label{Lem.mu}
Given any $R \in (0,\infty)$, $\mu_A(R) = 0$ if, and only if, 
$B=0$ on $D_R$.
\end{Lemma}
\begin{proof} 
The number~$\mu_A(R)$ as defined in~\eqref{mu} is the spectral threshold of 
the (self-adjoint) magnetic Schr\"odinger operator
$({\nabla_{\!A}})^{*} \, \nabla_{\!A}$ in $\sii(D_R)$,
subject to magnetic Neumann boundary conditions.
Since~$A$ is bounded on the disk~$D_R$,
the form domain coincides with the Sobolev space $W^{1,2}(D_R)$,
which is compactly embedded in $\sii(D_R)$.
Hence the spectrum of the operator is purely discrete,
the infimum in~\eqref{mu} is achieved 
and~$\mu_A(R)$ is the lowest eigenvalue of the operator. 
Let us denote by~$\psi_R$ a corresponding eigenfunction,
which is smooth due to elliptic regularity theory.

If $\mu_A(R)=0$, then $\nabla_{\!A}\psi_R=0$ in~$D_R$,
which is equivalent to the system of two equations
\begin{equation}\label{system}
\begin{aligned}
  -i \partial_1 \psi_R - A_1 \psi_R &= 0 \,,
  \\
  -i \partial_2 \psi_R - A_2 \psi_R &= 0 \,,
\end{aligned}
\end{equation}
where we write $A=(A_1,A_2)$.
Differentiating the first equation with respect to the second variable,
differentiating the second equation with respect to the first variable
and subtracting the two results, we arrive at the pointwise identity
\begin{equation}\label{pre-identity}
  B \, \psi_R = 0
  \quad \mbox{on} \quad
  D_R
  \,.
\end{equation}
At the same time,  
differentiating the first equation of~\eqref{system} 
with respect to the first variable,
differentiating the second equation with respect to the second variable
and adding the two results, we obtain that~$\psi_R$ satisfies 
the stationary Schr\"odinger equation    
$$
  -\Delta \psi_R + (i \divergence A - A^2) \psi_R = 0
  \quad \mbox{in} \quad
  D_R
  \,,
$$
where $A^2 := A_1^2+A_2^2$.
Taking into account that~\eqref{system} also implies 
that if $\psi_R$ vanishes at a point of~$D_R$
then it vanishes there of infinite order,
the unique continuation property ensures that~$\psi_R$
is nowhere zero in~$D_R$.
Hence, \eqref{pre-identity} implies that $B=0$ in~$D_R$. 

To prove the opposite implications, 
let us assume that $B=0$ in~$D_R$ (\ie~$A$ is closed). 
Then there exists a smooth function $F: \Real^2 \to \Com$
such that $A=\nabla F$ (\ie~$A$ is exact).
Choosing $\psi = e^{iF}$ as the test function 
in the variational definition~\eqref{mu},
we obtain $\mu_A(R)=0$.
\end{proof}

The following theorem is a variant of Theorem~\ref{Thm.Hardy.log} 
obtained under alternative hypotheses.
It is also a complex extension of the robust 
magnetic Hardy inequality~\cite[Thm.~1.1]{CK}.

\begin{Theorem}\label{Thm.Hardy.robust}
Let a smooth complex field $B:\Real^2\to\Com$ 
be not identically equal to zero. 
Then, for any smooth vector potential $A: \Real^2 \to \Com^2$
satisfying $B = \rot A$ and the hypothesis
\begin{equation}\label{Ass}
  \lim_{|x| \to \infty} 
  \big( |\Im A(x)| \, |x| \, \log(|x|) \big)
  = 0
  \,,
\end{equation}
there exists a positive constant~$\hat{c}$ 
depending on~$A$
such that the following inequality holds
\begin{equation}\label{Hardy.robust}
  \forall \psi \in C_0^\infty(\Real^2) \,, \qquad
  \int_{\Real^2} |\nabla_{\!A}\psi(x)|^2 \, \der x
  \geq \hat{c} \int_{\Real^2} \frac{|\psi(x)|^2}{1+|x|^2\log^2(|x|)} \, \der x
  \,.
\end{equation}
\end{Theorem}
\begin{proof} 
We are inspired by the gauge-free proof of~\cite[Thm.~1.1]{CK}. 
 
Let $\psi \in C_0^\infty(\Real^2)$.
The definition~\eqref{mu} yields
\begin{equation}\label{Hardy.free}
  \int_{\Real^2} |\nabla_{\!A}\psi(x)|^2 \, \der x
  \geq
  \mu_A(R)
  \int_{D_R} |\psi(x)|^2 \, \der x
  \,,
\end{equation}
where $\mu_A(R)$ is positive for all sufficiently large~$R$
whenever~$B$ is non-trivial.
This is a variant of the local Hardy inequality~\eqref{step1} 
with $I:=(0,R)$.
To conclude with the global Hardy inequality~\eqref{Hardy.robust} 
by mimicking the rest of the proof of Theorem~\ref{Thm.Hardy.log},
however, we would also need~\eqref{step1.5}.
For purely real-valued vector potentials considered in~\cite{CK},
a variant of~\eqref{step1.5} 
(with the absolute value~$|\phi|$ instead of~$\phi$)
follows by the diamagnetic inequality \cite[Thm.~7.21]{LL}.
In our more general case, this inequality only yields
\begin{equation}\label{dia}
  \forall x \in \Real^2 \,, \qquad
  |\nabla_{\!A}\psi(x)| \geq |\nabla_{\!\Im A}|\psi(x)||
  \,,
\end{equation}
and we have to proceed differently.

Set $r_0 := R/2$, the middle point of the interval~$I$.
Suppose for a moment that $\psi \in C_0^\infty(\Real^2 \setminus \overline{D_{r_0}})$. 
Then, for every real constant~$\beta$, we have
\begin{eqnarray*}
  \lefteqn{
  0 \leq
  \int_{\Real^2} 
  \left|
  \nabla_{\! \Im A} |\psi(x)| 
  - \beta \, \frac{x}{|x|^2 \log(|x|/r_0)} \, |\psi(x)|  
  \right|^2 \der x 
  }
  \\
  &&= 
  \int_{\Real^2} 
  |\nabla_{\! \Im A} |\psi(x)| |^2 \, \der x
  + \beta^2 \! \int_{\Real^2} \frac{|\psi(x)|^2}{|x|^2 \log^2(|x|/r_0)}  \, \der x
  \\
  &&
  \qquad
  - \beta \! \int_{\Real^2} \frac{x\cdot \nabla |\psi(x)|^2}{|x|^2 \log(|x|/r_0)} \, \der x
  - 2\beta \! \int_{\Real^2} \frac{x\cdot \Im A(x)}{|x|^2 \log(|x|/r_0)} \, |\psi(x)|^2 
  \, \der x
  \\
  &&= 
  \int_{\Real^2} 
  |\nabla_{\! \Im A} |\psi(x)| |^2 \, \der x
  + (\beta^2 - \beta) \int_{\Real^2} \frac{|\psi(x)|^2}{|x|^2 \log^2(|x|/r_0)}  \, \der x
  - 2\beta \! \int_{\Real^2} \frac{x\cdot \Im A(x)}{|x|^2 \log(|x|/r_0)} \, |\psi(x)|^2 
  \, \der x
  \,.
\end{eqnarray*}
Here the second equality follows by an integration by parts.
Recalling~\eqref{dia} and denoting
$$
  k_R :=
  \sup_{|x| \geq r_0} 
  \big( |\Im A(x)| \, |x| \, \log(|x|/r_0) \big)
  \,,
$$
which tends to zero as $R \to \infty$ due to~\eqref{Ass},
we have thus obtained the lower bound
\begin{align}\label{exterior}
  \int_{\Real^2} 
  |\nabla_{\! A} \psi(x) |^2 \, \der x
  &\geq 
  ( -\beta^2+\beta-2 \, |\beta| \, k_R )
  \int_{\Real^2} \frac{|\psi(x)|^2}{|x|^2 \log^2(|x|/r_0)}  \, \der x
  \nonumber \\
  &= \gamma_R
  \int_{\Real^2} \frac{|\psi(x)|^2}{|x|^2 \log^2(|x|/r_0)}  \, \der x
  \qquad \mbox{with} \qquad
  \gamma_R := \left( \frac{1-2 \, k_R}{2} \right)^2
  \,.
\end{align}
Here the equality follows by the best choice $\beta:=(1-2k_R)/2$, 
assuming that~$R$ is so large that $k_R < 1/2$.

To employ the estimate~\eqref{exterior} 
for our test function $\psi \in C_0^\infty(\Real^2)$,
we use the cut-off~$\xi$
from the proof of Theorem~\ref{Thm.Hardy.log}.
Similarly as in~\eqref{trick}, we have
\begin{eqnarray*}
\lefteqn{
  \int_{\Real^2} 
  \frac{|\psi(x)|^2}{1+|x|^2\log^2(|x|/r_0)} 
  \, \der x
}
\\
  &&\leq
  2 \int_{\Real^2 \setminus D_{r_0}} 
  \frac{|\xi(|x|) \psi(x)|^2}{|x|^2\log^2(|x|/r_0)} 
  \, \der x
  + 2 \int_{D_{r_0}} 
  |\xi(|x|) \psi(x)|^2
  \, \der x
  +  2 \int_{\Real^2}
  |(1-\xi(|x|)) \psi(x)|^2 
  \, \der x
  \\
  &&\leq
  \frac{2}{\gamma_R}  
  \int_{\Real^2 \setminus D_{r_0}} 
  |\nabla_{\! A} (\xi(|x|)\psi(x)) |^2 \, \der x
  +  2 \int_{D_{r_0}} 
  |\psi(x)|^2 
  \, \der x
  +  2 \int_{D_R} 
  |\psi(x)|^2 
  \, \der x
  \\
  &&\leq
  \frac{4}{\gamma_R}  
  \int_{\Real^2 \setminus D_{r_0}}  
  |\xi(|x|)|^2 \, |\nabla_{\!A} \psi(x)|^2 
  \, \der x
  + \frac{4}{\gamma_R}  \int_{\Real^2 \setminus D_{r_0}}  
  |\xi'(|x|)|^2 \,  |\psi(x)|^2 
  \, \der x
  + 4 \int_{D_R} |\psi(x)|^2 \, \der x  
  \\
  &&\leq
  \frac{4}{\gamma_R}  \int_{\Real^2 \setminus D_{r_0}}   
  |\nabla_{\!A}\psi(x)|^2 
  \, \der x
  + \left( \frac{4}{\gamma_R} \, \|\xi'\|_\infty^2 + 4 \right)
  \int_{D_R}
  |\psi(x)|^2 
  \, \der x
  \,,
\end{eqnarray*}
where we have used~\eqref{exterior} in the second inequality. 
In analogy with~\eqref{step2},
we have therefore proved 
\begin{equation}\label{step2.robust}
  \int_{\Real^2} |\nabla_{\!A}\psi(x)|^2 \, \der x
  \geq \frac{\gamma_R}{4} \int_{\Real^2}
  \frac{|\psi(x)|^2}{1+|x|^2\log^2(|x|/r_0)} \, \der x
  - \left(\|\xi\|_\infty^2 + \gamma_R \right)
  \int_{D_R}
 |\psi(x)|^2 \, \der x
  \,.
\end{equation}
As in the end of the proof of Theorem~\ref{Thm.Hardy.log},
interpolating between~\eqref{Hardy.free} and~\eqref{step2.robust},
we get
$$
  \int_{\Real^2} |\nabla_{\!A}\psi(x)|^2 \, \der x
  \geq
  \left[
  (1-\eps) \mu_A(R) - \eps \left(\|\xi\|_\infty^2+ \gamma_R\right)
  \right]
  \int_{D_R} |\psi(x)|^2 \, \der x
  + \eps \, \frac{\gamma_R}{4} \int_{\Real^d}
  \frac{|\psi(x)|^2}{1+|x|^2\log^2(|x|/r_0)} \, \der x
$$
with any $\eps > 0$.
Choosing~$\eps$ in such a way that the square bracket vanishes,
we obtain~\eqref{Hardy.robust} with
$$
  \hat{c} \geq
  \frac{\displaystyle \frac{\gamma_R}{4} \, \mu_A(R)}
  {\displaystyle \mu_A(R)+\|\xi\|_\infty^2+ \gamma_R} \
  \inf_{r \in (0,\infty)} \frac{1+r^2\log^2(r)}{1+r^2\log^2(r/r_0)}
  >0
  \,.
$$
The theorem is proved.
\end{proof}
\begin{Remark}
The hypothesis~\eqref{Ass} is clearly a condition on 
the rate of decay of~$\Im A$ at infinity. 
Such kind of condition would not be acceptable 
for the real part of~$A$ (because of the gauge invariance),
but it seems to be a relevant assumption 
when the imaginary part of~$A$ is present.
On the other hand, it is just a \emph{sufficient} condition 
to make the proof of Theorem~\ref{Thm.Hardy.robust} work.
An alternative to~\eqref{exterior} is to integrate by parts 
and get the identity
$$
\begin{aligned}
  \int_{\Real^2} |\nabla_{\! \Im A}\psi(x)|^2 \, \der x
  &= \int_{\Real^2} |\nabla|\psi(x)||^2 \, \der x
  + \int_{\Real^2} |\Im A(x)|^2 \, |\psi(x)|^2 \, \der x
  + \int_{\Real^2} \Im A(x) \cdot \nabla|\psi(x)|^2 \, \der x
  \\
  &= \int_{\Real^2} |\nabla|\psi(x)||^2 \, \der x
  + \int_{\Real^2} 
  \left(
  |\Im A(x)|^2 - \divergence \Im A(x)
  \right) 
  \, |\psi(x)|^2 \, \der x
\end{aligned}
$$
valid for all $\psi \in C_0^\infty(\Real^2)$.
If $|\Im A|^2 - \divergence \Im A$ is non-negative,
then one immediately gets~\eqref{step1.5} 
with help of~\eqref{dia} and one can proceed
as in the proof of Theorem~\ref{Thm.Hardy.log}
(with~\eqref{Hardy.free} instead of~\eqref{step1}). 
In fact, it is enought to assume that 
$|\Im A|^2 - \divergence \Im A$ is non-negative
outside the disk~$D_R$ where~$R$ can be chosen arbitrarily large.
Yet another possibility is to replace the sign condition
by the asymptotic behaviour 
\begin{equation*}
  \lim_{|x| \to \infty} 
  \left[
  ( |\Im A|^2 -\divergence \Im A)_-(x)
  \, |x|^2 \, \log^2(|x|) 
  \right]
  = 0
  \,,
\end{equation*}
where~$f_-$ denotes the negative part of a function~$f$. 
\end{Remark}
%


\subsection*{Acknowledgment}
%
The author was partially supported 
by the GACR grant No.\ 18-08835S
and by FCT (Portugal)
through project PTDC/MAT-CAL/4334/2014.

%
\bibliography{bib}
\bibliographystyle{amsplain}

\end{document}